\documentclass[12pt, dvipsnames]{amsart}
\usepackage{graphicx}
\usepackage[round]{natbib}
\usepackage[utf8]{inputenc}
\usepackage{enumerate}
\usepackage{euscript}
\usepackage{amsmath,amsthm,amsfonts, amssymb, phonetic}
\usepackage{amsaddr,amsmath}
\usepackage{tikz}
\usetikzlibrary{positioning}
\usepackage{tikz-3dplot} 
\usetikzlibrary{arrows, patterns, decorations.pathmorphing,backgrounds,positioning,fit,petri}
\usepackage{pgfplots}
\usepackage{pgfplotstable,textcomp}
\usepackage{dirtytalk,stmaryrd,ulem}
\def\citeapos#1{\citeauthor{#1}'s (\citeyear{#1})}

\newdimen\slantmathcorr
\def\oversl#1{
\setbox0=\hbox{$#1$}
\slantmathcorr=\wd0
\hskip 0.2\slantmathcorr \overline{\hbox to 0.8\wd0{%
\vphantom{\hbox{$#1$}}}}
\hskip-\wd0\hbox{$#1$}
}

\def\undersl#1{
\setbox0=\hbox{$#1$}
\slantmathcorr=\wd0
\underline{\hbox to 0.8\wd0{%
\vphantom{\hbox{$#1$}}}}
\hskip-0.8\wd0\hbox{$#1$}
}

\makeatletter
\renewcommand{\email}[2][]{%
  \ifx\emails\@empty\relax\else{\g@addto@macro\emails{,\space}}\fi%
  \@ifnotempty{#1}{\g@addto@macro\emails{\textrm{(#1)}\space}}%
  \g@addto@macro\emails{#2}%
}
\makeatother

\newtheorem{lemma}{Lemma}
\newtheorem{prop}{Proposition}

\newtheorem{claim}{Claim}

\newtheorem{cor}{Corollary}
\newtheorem{defn}{Definition}

\newcommand{\R}{\mathbb{R}}

\newcommand{\argmax}{\operatornamewithlimits{argmax}}
\newcommand{\Comp}{\operatornamewithlimits{\varbigcirc}}

\newcommand{\gB}{\mathcal B}

\newcommand{\gR}{\mathcal R}
\newcommand{\gF}{\mathcal F}

\newcommand{\gI}{\mathcal I}

\title{An Algebraic Approach to Revealed Preference}
\thanks{\footnotesize A previous version of this paper was circulated under the title \say{A Functional Approach to Revealed Preferences.} We gratefully acknowledge Marco Castillo, Christopher Chambers, Thomas Demuynck, Federico Echenique, and Matthew Pollison for useful comments.}

\author[Freer]{Mikhail Freer}
\address[Freer]{University of Essex}
\email[Freer]{m.freer@essex.ac.uk}

\author[Martinelli]{C{\'e}sar Martinelli}
\address[Martinelli]{George Mason University}
\email[Martinelli]{cmarti33@gmu.edu}

\begin{document}

\maketitle

\begin{abstract}
We propose and develop an algebraic approach to revealed preference.  Our approach dispenses with non algebraic structure, such as topological assumptions.  We provide algebraic axioms of revealed preference that subsume previous, classical revealed preference axioms, as well as generate new axioms for behavioral theories, and show that a data set is rationalizable if and only if it is consistent with an algebraic axiom. 
\end{abstract}

\section{Introduction}

The revealed preference approach to consumer choice, pioneered by \cite{samuelson1938note}, builds on the fact that, although we cannot observe the complete preference relation profiles of economic agents, we can observe their choices over some budget sets. Starting with the work \cite{richter1966} and \cite{afriat1967}, this approach has been used to construct tests of rational decision making \cite[see][for a recent comprehensive overview]{chambers2016revealed}.

\subsection*{Contribution}	
We propose an algebraic version of revealed preference approach.
That is, we consider theories about preferences in their logical structure together with the underlying algebraic structure.
By a theory about preferences, we mean a statement about preferences such as ``If $x$ is better than $y$ then $f(x)$ better than $f(y)$, where $f\in \gF$.'' In this statement, $f$ is some function over alternatives, and $\gF$ is a family of functions that actually defines the theory. The algebraic structure we impose considers the algebra of $(\gF,\circ)$, where $\circ$ is the composition operator.
In particular, we propose an \textit{algebraic axioms of revealed preferences}, and show that if $(\gF,\circ)$ is a group,\footnote{
	A tuple $(\gF,\circ)$ is said to be group if the set of functions $\gF$ contains an identify function,   $\gF$ is closed, every function in $\gF$ has an inverse that also belongs to $\gF$, and the composition operator $\circ$ is associative.
} 
then the observed set of data could be generated by a binary relation consistent with the theory $\gF$ if and only if it is consistent with the algebraic axiom of revealed preferences.

The theory as defined does not include the assumption of transitivity, which is considered separately. Hence, we define the weak and strong algebraic axioms  similarly to the weak and strong axioms of revealed preferences.   We show that the weak (algebraic) axiom is the equivalent to the possibility that the data could be generated by a binary relation consistent with the theory $\gF$; while the strong (algebraic) axiom is equivalent to the possibility that the data could be generated by a transitive binary relation consistent with the theory $\gF$.   Moreover, we show that if we do not require  the preference relation to be transitive, we obtain completeness for free.  That is, a data could be generated by a binary relation consistent with theory $\gF$ if and only if it could be generated by a complete binary relation consistent with the theory $\gF$.	If we do require the preference relation to be transitive, in addition to being consistent, we need to add a mild assumption in order to get the complete rationalization.
	
We show that our result unifies classical results \citep[see][]{afriat1967,varian1983} and more recent contributions \citep[see][]{forges2009,heufer2013testing,nishimura2017comprehensive,polisson2015,castillofreer2016}  in revealed preference theory.
In particular, we provide applications to the existence of transitive, homothetic, and quasilinear preferences, as well as preferences that satisfy independence. Importantly, the latter example includes both objective and subjective probability versions.
Finally, we consider new applications to the behavioral theories of sequential choice procedures and good enough choice.
Like older rational choice approaches, both of these theories require that observed choices are generated  by a well-behaving binary relation.
	We show that, with the help of our general result,  standard classical revealed preference tests can be extended to incorporate these theories.

\subsection*{Related literature}
	Our work is linked to the literature on generalized revealed preferences.
	Several authors in this literature provide a generalization of the revealed preference approach, but keeping some topological assumptions in place. Topological assumptions are necessary to guarantee existence of a convenient utility function representing the underlying preference relation.	Seminal examples are \cite{chavas}, \cite{forges2009}, and \cite{nishimura2017comprehensive}, who generalize  \cite{afriat1967} theorem for general shapes of budgets and topological spaces. Recently, \cite*{polisson2015} have proposed a lattice approach, and provided conditions that guarantee the rationalization of an observed set of data with theories such as expected utility, ranked dependent expected utility, and cumulative prospect theory.
	The papers mentioned above construct tests which can be easily applied to the data.

	Other authors generalize \citeapos{szpilrajn1930} result, concentrating on the completion of the revealed preference relation. Seminal papers by \cite{suzumura1976remarks}, \cite{duggan1999general}, and \cite{demuynck2009} provide revealed preference tests (in the shape of \textit{Suzumura consistency}) for transitive, acyclic, homothetic, and convex preferences.
	However, the Suzumura consistency condition may be complicated for practical implementation, which is an important difference with the papers mentioned in the previous paragraph.  
	Let us also note that Suzumura consistency is not a revealed preference axiom in the orthodox sense, as it is not stated in terms of choices and budgets.
	In this sense, we provide a link between these two strands of revealed preference research.  That is, we adopt a scope of theories comparable to one presented in \cite{demuynck2009}, while providing a tractable and simple revealed preference axiom.

	Unlike most of the previous literature, we do not require the revealed preference relation to be generated by a complete preference relation.    As we know from \cite{chambers2014axiomatic}, completeness does not bring additional empirical content if transitivity is required.   Our results in terms of adding completeness to other desiderata can be interpreted in that spirit.  That is, instead of taking completeness as a desiderata,  we focus on whether we can get it and at what costs.
	Sometimes, it appears that completeness is obtained for free, in the sense that adding the assumption of completeness to other assumptions about the binary relation does not bring any additional empirical content. 

\subsection*{Organization of the paper}
	The remainder of the paper is organized as follows.
	We present the necessary definitions for algebraic revealed preferences in section 2.
	We show our rationalizability result in section 3.
	We present applications in Section 4.
	We provide some concluding remarks in section 5.
	All proofs omitted in the text are collected in an Appendix.

\section{Preliminaries}
\noindent
Let $X$ be the space of alternatives.  
Let $R\subseteq X\times X$ be a \emph{preference relation}, that is a reflexive and transitive binary relation. 
A binary relation is said to be \emph{reflexive} if  $(x,x)\in R$ for every $x\in X$. 
A binary relation is said to be \emph{transitive} if $(x,y) \in R$ and $(y,z)\in R$ implies $(x,z)\in R$ for every $x,y,z\in X$. 
A binary relation is said to be \emph{complete} if all pairs are comparable, i.e. $(x,y)\in R$ or $(y,x)\in R$ for every $x,y\in X$.
Let $P(R)$ is the (asymmetric) strict part of the relation; that is, $P(R) = \{(x,y)\in R: \ (y,x)\notin R\}$.

\subsection{Theories about preferences.}
Next we present our notion of a theory about preference relations.  Our notion is meant to capture the idea that a theory is a collection of desirable  properties of preference relations.
We first discuss some usual theories to motivate our notion.
\newline

\noindent
{\bf Example:}
Consider a theory that imposes \emph{homothetic preferences}, that is $(x,y)\in R$  implies $(\alpha x, \alpha y)\in R$ for every $\alpha\in \R_{++}$.
\newline

\noindent
{\bf Example:}
Consider a theory that imposes \emph{quasilinear preferences}, that is $(x,y)\in R$ implies $(x + \alpha e, y + \alpha e)\in R$ for every $\alpha\in \R$, where $e =(1,0,\ldots,0,\ldots,0)$ is the vector of zeros with unique 1 element.
\newline

\noindent
We can generalize the structure observed in both of these theories as follows.
Both $\alpha x$ and $x+\alpha e$ can be in general presented as functions $f: X\rightarrow X$.
Then, the theory itself can be described as a collection of functions $\gF$ preserving the preference relation.
For homothetic preferences the allowed functions would be all linear functions with positive slope and zero intercept.
For quasilinear preferences the allowed functions would be all linear functions with slope of one and intercept of $\alpha e$ for every real $\alpha$.
We generalize this idea allowing for other collections of functions, so that every allowed collection of functions $\gF$ defines a theory.
In particular, we require the collection of functions $\gF$ endowed with the composition operator to be a group.

\begin{defn}
A tuple $(\gF,\circ)$ is a {\bf group} if it
	\begin{itemize}
		\item contains identity:
		\\
		$I\in \gF$, where $I(x)=x$ for every $x\in X$;

		\item is closed: 
		\\ $f,f'\in \gF$ implies $f\circ f' = f'' \in \gF$;
		\item has inverse element: \\
		$\forall f\in \gF \ \exists f^{-1}\in \gF$ such that $f\circ f^{-1} = f^{-1}\circ f = I$;
		\item is associative: 
		\\ 
		$(f\circ f')\circ f'' = f\circ(f'\circ f'')$ for all $f,f',f''\in \gF$.
	\end{itemize}
\end{defn}

As is well-known, the collection of all bijective functions from a set to itself constitutes a group.  
Thus, the assumption of group structure is not very restrictive once we can characterize a theory as a collection of transformations.  
The requirement of identify guarantees that applying theory we are not loosing information about preferences.  
Closedness implies theory is self-contained, without this requirement the theory may generate some implications which are not directly prescribed by this theory.
Associativity is a technical property of the composition operator.

The requirement of inverse element is more substantial; intuitively, it means that if it is desirable that $(x,y) \in R$ implies $(f(x),f(y)) \in R$, then the reverse implication is desirable as well.  This inverse implication is actually guaranteed by existence of the inverse element for every function $f\in \gF$. As we want to consider not only the complete relations being consistent with the theory, the reverse implication is crucial. For the example of its usefulness, one can see \cite{aumann1962utility} where it is illustrated how the independence axiom needs to be modified for the representation of incomplete preference relation satisfying transitivity.

\begin{defn}
Let $(\gF,\circ)$ be a group.
A preference relation $R$ is {\bf consistent with theory $\gF$} if 

\begin{itemize}
	\item [] $(x,y) \in R$ implies $(f(x),f(y))\in R$ for every $f\in\gF$.
\end{itemize}

\end{defn}

\noindent
Note that our notion of theory does not involve completeness or transitivity. Neither of these assumptions can be expressed in the simple terms of the theory as we define it.  Completeness and/or transitivity will be considered as additional assumptions when introducing the notions of rationalization of the data set with a preference relation. This also allows us to consider the interaction between these assumptions and our notion of the theory.

\subsection{Data and Rationalization.}
The essence of the revealed preference problem is to extract the (unobserved) preference relations which generated (observed) choices over budgets.
If there is such a preference relation, then the corresponding data set (collection of choices from budgets) is rationalizable.
Next, we formally define the rationalizable data sets.

Let $B\subseteq X$ be a budget set, where $B$ is any nonempty subset of $X$.
Let $\gB$ be a collection of budgets.
	Let $C:\gB\rightrightarrows X$ be a choice correspondence.
Denote by $(\gB,C)$ a data set.	Let $x\in \max (B,R)$ be a \emph{maximal point} of budget set $B$ according to preference relation $R$ if $(x,y)\in R$ for every $y\in B$ and $(y,x)\notin P(R)$.\footnote{
		Our notion of maximal point (even being standard for revealed preferences) already implicitly makes an assumption about comparability.
		That is, the chosen point is at least as good as any other point in the budget, and not just an undominated point.
		However, relaxing it would result in the absence of  empirical content of any theory that does not explicitly assumes completeness.
	}
	Rationalizability requires that the chosen points are maximal within the corresponding budget for a preference relation that is consistent with the theory $\gF$. 
	
	We present four different versions of rationalizability, according to whether they include or not transitivity and completeness as desiderata.  Rationalizability without transitivity is denoted `weak' and rationalizability with completeness is denoted `complete.'
\begin{defn}
\label{def:AllRationalizability }
A data set $(\gB,C)$ is 
\begin{itemize}
	\item [] {\bf weakly rationalizable} if there is a preference relation $R^*$ consistent with theory $\gF$,

	\item [] {\bf completely weakly rationalizable} if there is a complete preference relation $R^*$ consistent with theory $\gF$,

	\item [] {\bf rationalizable} if there is a transitive preference relation $R^*$ consistent with theory $\gF$,

	\item [] {\bf completely rationalizable} if there is a transitive and complete preference relation $R^*$ consistent with theory $\gF$,
\end{itemize}
such that 
$$
C(B) = \max (B,R^*) \text{ for every } \ B\in \gB.
$$
\end{defn}

\noindent
The definition of rationalizability specifies that there is a preference relation consistent with a given theory, such that the observed choices can be generated by maximization of this preference relation.  
Notably, we abstain from any restrictions on the cardinality of the data set and any topological structure of the space of alternatives, since we are pursuing a purely algebraic approach.  
A disadvantage of this level of generality is that we cannot talk about a utility function representing the preference relation.

\section{Results}
\label{sec:Results}
\noindent
In what follows, we establish the equivalence between four versions of the axiom of revealed preference and the four notions of rationalization according to $\gF$.  In presenting some of the results, we denote the composition of $f_1,\ldots, f_n\in \gF$  by 
$$
\Comp_{j=1}^n f_j = f_1\circ f_2\circ \ldots \circ f_n.
$$
	
\subsection{Non-complete rationalization}
Following the revealed preference tradition, we define two axioms: weak and strong.  Parallel to standard results in revealed preference theory, the weak axiom does not account for transitivity, while the strong one does.

\begin{defn}
\label{def:WAARP}
A data set satisfies the {\bf weak algebraic axiom of revealed preference (WAARP)} if for every $x_i \in C(B_i)$ and $x_j\in C(B_j)$ and every $f\in \gF$
$$
f(x_i)\in B_j \text{ implies } f^{-1}(x_j)\notin B_i\setminus C(B_i).
$$
\end{defn}

\noindent
Recall that the standard weak axiom requires that if $x_j\in C(B_j)$ is available at the budget $B_i$ then $x_i\in C(B_i)$ should not be available at $B_j$.  Otherwise, there would be a contradiction between $x_j$ being better than $x_i$ and $x_j$ being better than $x_i$.  Since now we incorporate the theory $\gF$, we also need to make sure that there is no $f\in \gF$ such that $x_j$ is better than $f(x_i)$, and $f(x_i)$ is better than $x_j$.  This reasoning illustrates the necessity of the weak axiom.  

\begin{prop}
\label{prop:WeakRationalization}
A data set is weakly rationalizable if and only if it satisfies WAARP.
\end{prop}

\noindent
	In parallel to the standard logic behind the strong axiom, we need to account for the indirect preference relation. This is done via replacing a single point $x_j\in B_i$ by the sequence of $x_n\in B_{n-1}, x_{n-1}\in B_{n-1},\ldots, x_{2}\in B_1$, which would imply that $x_n\notin B_1$. The existence of the sequence described above would ensure that $x_1$ is better than $x_n$ if preference relation is transitive, while the conclusion guarantees that $x_n$ is not better than $x_1$. 

\begin{defn}
A data set $(\gB,C)$ satisfies the {\bf strong algebraic axiom of revealed preference (SAARP)} if for every sequence $x_1, \ldots, x_n$ such that $x_j\in C(B_j)$ for every $j\in\{1,\ldots, n\}$  for some $B_1,\ldots, B_n\in \gB$, and every sequence $f_1,\ldots,f_{n-1}\in \gF$ such that
$$ 
f_{j}(x_{j+1}) \in B_{j+1} \text{ for every } j\in\{1,\ldots, n-1\},
$$ 
we have $$ \left[\Comp\limits_{j=1}^{n-1} f_j\right]^{-1} (x_1) \notin B_n\setminus C(B_n).
$$
\end{defn}

\noindent
The strong axiom should also incorporate the transformations allowed by the theory.  Let us provide some intuition for the axiom.  Assuming that $f_{j}(x_{j+1}) \in B_j$ implies that $x_j$ is better than $f_j(x_{j+1})$.  Since we require our preference relation to be consistent with $\gF$, then $x_j$ is better than  $f_j(x_{j+1})$ , which implies that $f_{j-1}(x_{j-1})$ is better than $[f_{j-1}\circ f_j](x_{j+1})$. Following this logic sequentially we obtain that 
	$$
	x_1  \text{ is better than } \left[\Comp\limits_{j=1}^{n-1} f_j\right] (x_n).
	$$
	Hence, we need to ensure that 	
	$$
	 \left[\Comp\limits_{j=1}^{n-1} f_j\right] (x_n) \text{ is not better than } x_1.
	$$
	Given that the preference relation is consistent with $\gF$, the claim above is equivalent to taking the inverse of the function from both sides, i.e.
	$$
	x_n \text{ is not better than } \left[\Comp\limits_{j=1}^{n-1} f_j\right]^{-1} (x_1).
	$$
	While this argument illustrates the necessity of SAARP, the axiom is also sufficient for rationalization.

\begin{prop}
\label{prop:Rationalization}
A data set is rationalizable if and only if it satisfies SAARP.
\end{prop}

\subsection{Complete Rationalization}
Our previous results for rationalization include as desiderata only consistency with the theory and transitivity. We include now completeness of the preference relation.  As it happens, there is no need to refine the conditions to obtain complete weak rationalization and the very same weak axiom is necessary and sufficient.

\begin{prop}
\label{prop:CompleteWeakRationalization}
A data set is completely weakly rationalizable if and only if it satisfies WAARP.
\end{prop}

\noindent
The idea of our proof is to provide a (rather arbitrary) algorithm which completes the observed revealed preference relation in a way that it is consistent with the theory.   While in the case of Proposition \ref{prop:CompleteWeakRationalization} such algorithm clearly converges, once we add the transitivity convergence cannot be guaranteed without additional structure.   The additional structure incorporates notions of order and monotonicity illustrated by the following examples:

\bigskip

\noindent
{\bf Example:}
Recall a theory that imposes homothetic preferences, that is $(x,y)\in R$  implies $(\alpha x, \alpha y)\in R$ for every $\alpha\in \R_{++}$.
Note that the functions are completely ordered, that is for every $\alpha,\alpha'\in\R_{++}$ either $\alpha\geq \alpha'$ or $\alpha'\geq \alpha$, where $\geq$ is the standard greater or equal order on $\R_{++}$.
Moreover, it makes sense to assume that homotheticity would also imply that $(\alpha x, \alpha' x)\in R$ if $\alpha\geq \alpha'$.

\bigskip

\noindent
{\bf Example:}
Recall a theory that imposes quasilinear preferences, that is $(x,y)\in R$ implies $(x + \alpha e, y + \alpha e)\in R$ for every $\alpha\in \R$, where $e =(1,0,\ldots,0,\ldots,0)$ is the vector of zeros with unique 1 element.
Note that the functions we consider are completely ordered, that is for every $\alpha,\alpha'\in\R$ either $\alpha\geq \alpha'$ or $\alpha'\geq \alpha$, where $\geq$ is the standard greater or equal order on $\R$.
Similarly to homotheticity, it makes sense to assume the monotonicity with respect to $\alpha$.
That is $(x + \alpha e, x + \alpha' e)\in R$ if $\alpha\ge \alpha'$.

\bigskip

\noindent
Given a group  $(\gF,\circ)$, a triple $(\gF,\circ,\ge)$ is an \emph{ordered group} if  $\geq$ is a complete order such that
	\begin{align*}
	& f\ge f' \text{ implies } f''\circ f \ge f'' \circ f' \\
	& f\ge f' \text{ implies } f\circ f'' \ge f' \circ f''
	\end{align*}
	for every $f,f',f''\in \gF$.
	Note that we have to use both left- and right-ordered assumptions as we did not assume that the $\circ$ operator is commutative.
	A theory $\gF$ is \emph{ordered} if $(\gF,\circ,\ge)$ is an \emph{ordered group}.
	
	We need to modify the notion of consistency in order to take into account the that the theory also imposes monotonicity with respect to $\geq$ over its own functions.
	Such modification of the consistency is necessary, as otherwise there is no value for this extra structure incorporated, as it is not being reflected in the theory itself.
	A binary relation $R$ is \emph{consistent with ordered theory $\gF$} if 
	$$
	(x,y)\in R \text{ implies } (f'(x),f(y))\in R \text{ for every } f'\geq f\in \gF.
	$$
	Interestingly, the order imposed on the theory is passed to the preferences.  Moreover, the completeness assumption on the order is in some way passed as well.
	Note that introducing a partial order over functions would not be enough. Our extra assumption is congruent with \cite{chambers2014axiomatic} result about the absence of the empirical content of completeness, since a trivial group containing only the identity operator is obviously ordered.

	Finally, given the refinements of the theory, we need to introduce some additional constraints on the data set.
	A data set $(\gB, C)$ is \emph{regular} if 
	\begin{itemize}
		\item [--] $f(x) \in B$ for some $B\in \gB$, then $\bar f, f\in B$ for every $\bar f\in \gF$ such that $\bar f\le f$ and every $x\in X$,
		\item [--] $f(x)\in C(B)$ for some $B\in B$, then $\bar f(x)\notin B$ for every $\bar f,f \in \gF$ such that $\bar f\le f$ and every $x\in X$.
	\end{itemize}
	Even though we introduce regularity as assumption, it can be incorporated into the testable conditions.  The condition on budgets corresponds to the notion of downward closure of the budget frequently employed in revealed preference literature.  The condition on the chosen point refers to the fact that chosen point should be on the border of the budget set.

\begin{prop}
\label{prop:CompleteRationalization}
Let $\gF$ be a completely ordered theory and $(\gB,C)$ be a regular data set.
A data set is completely rationalizable if and only if it satisfies SAARP.
\end{prop}

\section{Applications}
\label{sec:Applications}

\noindent
We consider two sets of applications.
First, we consider classical theories about preferences, in order to illustrate the basic scope of the algebraic approach.
Second, we consider several behavioral theories, focusing on theories of limited attention. These applications merge  behavioral considerations (such as limited attention) with classical desiderata (as homotheticity, quasilinearity, or independence).

\subsection{Classical Applications}
In this section we show that classical theories of preferences fit in the framework. 
Note that for some of the theories we have to introduce some extra structure over the space of alternatives which is not necessary for the general result, but unavoidable once we want to define the particular theory.  
We focus on rationalization and complete rationalization, as for classical applications, transitivity is commonly considered as a requirement.  
One can simplify these axioms down to their weak counterparts, by considering sequences of the length no more than two.

Before we proceed further, let us refresh the reader on some definitions in abstract algebra. The definition of a field resembles the standard number systems we are used to.	A \emph{field} denoted by $(A,+,*)$ is a set $A$ endowed with two operations ``addition'' ($+$) and ``multiplication'' ($*$) which contains to special elements zero ($0$) and one ($1$).	Also, every element has an inverse with respect to addition denoted by $-a$ and every non-zero element has an inverse with respect to multiplication denoted by $a^{-1}$.
	``Zero'' is an element such that $a+(-a) = 0$ for every $a\in A$.
	``One'' is an element such that $a*a^{-1} = 1$ for every $a\in A\setminus\{0\}$.
	In addition, the operations satisfy the following properties.
	\begin{itemize}
		\item [--] Associativity
		$$
		a+(b+c) = (a+b)+c \text{ and } a*(b*c) = (a*b)*c.
		$$

		\item [--] Commutativity
		$$
		a+b = b+a \text{ and } a*b = b*a.
		$$

		\item [--] Identity
		$$
		a+0 = a \text{ and } a*1 = a.
		$$

		\item [--] Distributivity
		$$
		a*(b+c) = (a*b) + (a*c).
		$$
	\end{itemize}

\noindent
An \emph{ordered field} $(A,+,*,\ge)$ is a field $(A,+,*)$ endowed with a total order $\geq$.  We assume two key properties.
First, adding the same number to both sides does not change the order, i.e. $a\ge b$ implies $a+c\ge b+c$.
Second, multiplying two positive elements we obtain positive elements, i.e. $a,b\ge 0$ implies $a*b \ge 0$.
It happens that these two properties buy us a lot, including the possibility of adding inequalities, and of multiplying inequalities by a \say{positive} number.  We also get the centralization around zero, that is either $a\ge 0\ge -a$ or $-a\ge 0\ge a$ and around one, that is $a\ge b$ where $a,b>0$ implies $b^{-1}\ge a^{-1}$.
In general, when working with an ordered field, one can intuitively apply the logic of real line since $\R$ is an ordered field.

\subsubsection*{Transitive Preferences}
\label{par:transitive}
Let $X$ be a universal set of alternatives.

\begin{defn}
A preference relation is said to be {\bf transitive}  if 
$$
(x,y) \in R \text{ and } (y,z)\in R \text{ implies } (x,z) \in R.
$$
for every $x,y,z\in X$.

\end{defn}

\noindent
Since transitivity is already embedded in the notion of preference relation, and therefore in the notion of rationalization, a statement of the theory in terms of functions can look simply as
	$$
	(x,y) \in R \text{ implies } (x,y)\in R.
	$$ 
	Hence, we can define 
$$
\mathcal{T}=\{I\}.
$$ 
Trivially, $(\mathcal{T},\circ)$ is a group, since it contains a unique element that is the identity function.
Moreover, it can be easily seen that $\mathcal{T}$ is the correct theory describing transitive preferences given the definition above.
Moreover, AARP in this case is equivalent to SARP.

\begin{defn}
A data set $(\gB,C)$ satisfies the {\bf Strong Axiom of Revealed Preference (SARP)} if for every sequence $x_1, \ldots, x_n$ such that $x_j\in C(B_j)$ for every $ j\in\{1,\ldots, n\}$  for some $B_1,\ldots, B_n\in \gB$, if
$$
x_{j+1} \in B_j \text{ for every }  j\in\{1,\ldots, n-1\}$$ then $$ x_1 \notin B_n\setminus C(B_n).
$$
\end{defn}

  $(\mathcal T, \circ,\ge)$ is trivially an ordered group as it contains single element.
	Therefore, the SARP is equivalent to the (complete) rationalization by means of the main results.

\begin{cor}
\label{cor:TransitiveApplication}
A data set is completely rationalizable with transitive preferences if and only if it satisfies SARP
\end{cor}

\subsubsection*{Homothetic Preferences.}
\label{par:HomotheticApplication}

Let $X$ be a vector space over a fully ordered field $(A,*,+\geq)$, and denote by $A_+$ the subspace of $A$ such that all $\alpha\in A$ is such that $\alpha\geq 0$.

\begin{defn}
A preference relation is said to be {\bf homothetic} if
	$$
	(x,y) \in R \text{ implies  } (\alpha x,\alpha y) \in R
	$$
	for every $x,y\in X$ and $\alpha \in A_+$.

\end{defn}

We can define the  
$$
\mathcal{H} = \{f(x) = \alpha x: \alpha\in A_+\}.
$$  
It is easy to see that $(\mathcal H,\circ)$ is a group.
Moreover since $A_+$ is an ordered field as well, then $(\mathcal H, \circ, \geq)$ is an ordered group.
In this case AARP is equivalent to a homothetic axiom of revealed preferences that generalizes the one proposed by \cite{varian1983}, \cite{heufer2013testing}, and \cite{heufer2019homothetic}.

\begin{defn}
A data set $(\gB,C)$ satisfies the {\bf Homothetic Axiom of Revealed Preference (HARP)} if for every sequence $x_1, \ldots, x_n$ such that $x_j\in C(B_j)$ for every $ j\in\{1,\ldots, n\}$  for some $B_1,\ldots, B_n\in \gB$, and every sequence $ \alpha_1\ldots,\alpha_{n-1}\in A_+$ such that
$$
\alpha_{j} x_{j+1}\in B_j \text{ for every } j\in\{1,\ldots, n-1\},
$$ 
we have 
$$ \frac{x_1}{\Pi_{j=1}^{n-1} \alpha_j}\notin B_{n}\setminus C(B_n).$$
\end{defn}

Note that HARP uses division as the operator being equivalent to the multiplication by $a^{-1}$.
	Since $(\mathcal H,\circ,\ge)$ is an ordered group, HARP is equivalent to the (complete) rationalization by the means of the main results, if we incorporate monotonicity into the scope of the theory.

\begin{cor}
\label{cor:HomotheticApplication}
A data set is completely rationalizable with homothetic preferences if and only if it satisfies HARP
\end{cor}

\subsubsection*{Quasilinear Preferences}
\label{par:QuasilinearApplication}
Let $X$ be a vector space over a fully ordered field $(A,+,*,\ge)$.
Denote by $e=(1,0,0,\ldots)$ the vector with unique $1$ element in the first place and zeros elsewhere, since we can always reorder the goods, such that the first good would be a numeraire.

\begin{defn}
A preference relation is said to be {\bf quasilinear} if
	$$
	(x,y)\in R \text{ implies } (x+\alpha e,y+\alpha e) \in R
	$$
	for every $x,y\in X$ and every $\alpha \in A$.
\end{defn}

We can define 
$$ 
\mathcal Q =  \{f(x) = x+\alpha e: \alpha\in A\}.
$$  
It is easy to see that $(\mathcal Q, \circ)$ is a group.
Moreover, ($\mathcal Q,\circ,\ge)$ is an ordered group.
In this case, AARP is equivalent to a quasilinear axiom of revealed preferences that generalizes the one proposed by \cite{rochet1987necessary}, \cite{brown2007nonparametric}, and \cite{castillofreer2016}.

\begin{defn}
A data set  $(\gB,C)$  satisfies the {\bf Quasilinear Axiom of Revealed Preference (QARP)} if for every sequence $x_1, \ldots, x_n$ such that $x_j\in C(B_j)$ for every $ j\in\{1,\ldots, n\}$  for some $B_1,\ldots, B_n\in \gB$, and every sequence $\alpha_1, \ldots,\alpha_{n-1} \in A $ such that
$$
x_{j+1}+ \alpha_{j} e \in B_j \text{ for every } j\in\{1,\ldots, n-1\},
$$ 
we have
$$ 
x_1- \sum_{j=2}^n \alpha_j \notin B_{n}\setminus C(B_n).
$$
\end{defn}

	Since $(\mathcal Q,\circ,\ge)$ is an ordered group, QARP is equivalent to the (complete) rationalization by the means of the main results, if we incorporate the monotonicity with respect to numeraire into the scope of the theory.

\begin{cor}
\label{cor:QuasilinearApplication}
A data set is completely rationalizable with quasilinear preferences if and only if it satisfies HARP
\end{cor}

\subsubsection*{Preferences Satisfying Independence}
\label{par:Independence}
	 
	Consider a set $\Omega$ of  potential outcomes, and an ordered field $A$ denoting the range for the measure function to be constructed. Since $A$ is a field, it has both $0$ and $1$, which correspond to outcomes happening with zero probability and happening for sure. 	Let $X$ be a vector space over $A$ with dimension of $X$ being equal to $\vert \Omega \vert$. That is, $X$ represents the probability distribution over the prizes in $\Omega$.	Note that $X$ can potentially be an infinitely dimensional space.

	Technically, we define the space of lotteries to be bigger than usual; properly, we should define it as a simplex over the space of outcomes. A problem which would appear in that case is that applying independence one can easily jump out of the standard simplex.	The difference between our definition and the standard one is similar to using  $\R^n$ instead of the simplex of the corresponding dimension.  That is, we construct a preference relation over a larger space but primarily for the matter of mathematical convenience and this preference relation can be truncated ex-post.  Let us note that it is possible to directly consider the simplex, however, it would require significant abuse of notation throughout the paper.\footnote{
		A way out is not to consider a theory as being a group of functions, but being \say{generated} by a group of functions.
		By being generated we mean that we can truncate both domain and image of the function to fit the set of alternatives.
	}
	

\begin{defn}
A preference relation is said to satisfy {\bf independence}
	if
	$$
	(x,y) \in R \text{ implies } (\alpha x + (1-\alpha) z, \alpha y + (1-\alpha)z ) \in R
	$$
	for every $x,y,z\in X$ and $\alpha \in A_{++}$.
\end{defn}

\noindent	
We use \citeapos{aumann1962utility} version of the independence axiom exactly because we have no completeness assumption, and we allow $\alpha$ to go above $1$. This statement of independence guarantees that the corresponding family of functions given the composition operator  form a group.

	We define 
	$$
	\gI = \{f(x) = \alpha x +  (1-\alpha)z: \alpha \in A_{++}; z\in X\}.
	$$
	Let us verify that  $(\gI,\circ)$ is a group, given that this example is less straightforward than previous examples.
	Denote $f(x) =\alpha x + (1-\alpha ) z$ by $f_{\alpha,z}(x)$.
	Hence, the inverse for this function is $f_{\frac{1}{\alpha},z}(x)$, that is also an eligible function.
	Composition of the functions is given by
	$$ 
	[f_{\alpha,z}(x) \circ f_{\alpha',z'}](x) = f_{\tilde\alpha,\tilde z}$$ 
	where 
	$$ 
	\tilde\alpha  = \alpha \alpha' \text{ and } \tilde z  =\frac{\alpha'(1-\alpha)}{1-\alpha'\alpha} z + \frac{1-\alpha'}{1-\alpha'\alpha} z'
	$$ 
	if $\alpha\alpha'\neq 1$.
	In this case, $\tilde\alpha \in A_{++}$, and $z \in X$ is the convex combination of two elegible lotteries.
	Note that when stating the notion of  theory we did not require the group $(\gF,\circ)$ to be commutative, i.e. it may be the case that $f\circ f' \neq f'\circ f$.
	Independence illustrates why we did not impose this condition, which would have simplified the exposition.
	In the case of independence, it is easy to see that commutativeness does not have to be satisfied.

	Next, we provide the specification of the AARP for the case of independence.
	The axiom can be linked to the tests of \cite{demuynck2009nash,polisson2015}.

\begin{defn}
A data set $(\gB,C)$ satisfies the {\bf Independence Axiom of Revealed Preference (IARP)} if for every sequence $x_1, \ldots, x_n$ such that $x_j\in C(B_j)$ for every $ j\in\{1,\ldots, n\}$  for some $B_1,\ldots, B_n\in \gB$, and every sequences $\alpha_1, \ldots,\alpha_{n-1} \in A_{++}$ and $z_1,\ldots, z_n \in X$ 
such that
$$
f_{\alpha_j,z_j}(x) \in B_j \text{ for every } j\in\{1,\ldots, n-1\},
$$ 
we have
$$
\left[\Comp\limits_{j=n-1}^{1} f_{\frac{1}{\alpha_j},z_j }\right] (x_1) \notin B_{n}\setminus C(B_n).
$$
\end{defn}

Note that IARP is only equivalent to rationalization. We cannot represent $\mathcal I$ as an ordered group, even though this group can be partially ordered. Note also that IARP is equivalent to complete rationalization as one can show that it is (Farkas) alternative to the conditions of \cite{polisson2015}. However, $\mathcal I$ introduces a lot of structure to the functions, i.e.\ requiring them to be linear. Studying theories induced by groups of linear functions can be a fruitful avenue for the further research.

\begin{cor}
\label{cor:IndependentApplication}
A data set is rationalizable with preferences satisfying independence if and only if it satisfies IARP
\end{cor}

\subsection{Behavioral Applications}
Further we provide application to the theories which can be associated with limited attention.
In particular, we consider theories of sequential choice as in \cite{manzini2013two} and good enough choice as in \cite{barbera2019order}.  Sequential procedures involve two preference relations---the first preference relation generates a consideration set, and the second generates a choice out of the consideration set formed at the first stage.
Under good enough choice, a choice is not the result of proper maximization but a selection of the good enough alternative.
That is, the consideration set may exclude some alternatives better than the chosen one.

\subsubsection*{Sequential Procedures}
\label{par:SequentialApplication}

The model for rationalization is adapted from \cite{manzini2013two}, and it is an example of sequential decision rules.
Unlike other models of sequential decision making \citep[see e.g.][]{manzini2007sequentially,manzini2012choice}, it assumes that all alternatives passed the preliminary stages are better than some threshold.  Next we provide the formal statement of the model.

Let $X$ be a universal set of alternatives.
Denote these two preference relations by $R_1$ and $R_2$ which correspond to the first and second stage of the decision-making process.
Let $\gF_1$ be a theory and $\gF_2$ be an ordered theory, as the requirement for the second stage preference relation would include both completeness and transitivity.

\begin{defn}
A data set $(\gB,C)$ is {\bf sequentially rationalizable} if there are
\begin{itemize}
	\item [--]  a binary relation $R_1$ consistent with theory $\gF_1$ and complete and
	\item [--] a transitive binary relation $R_2$  consistent with ordered theory $\gF_2$,
\end{itemize}  
such that
\begin{align*}
& \Gamma(B) = \max(B,R_1) \mbox{ and } \\
& C(B) = \max(\Gamma(B),R_2).
\end{align*}
\end{defn}

\noindent
In line with the general version,  sequential rationalization assumes that the chosen alternative is the best in the consideration set (according to $R_2$). With respect to the consideration set, we assume that all alternatives above a threshold are considered.  That is, there is a threshold alternative $y$ (or the entire equivalence class of such alternatives), that is the worst (according to $R_1$) alternative considered.  Every alternative considered is at least as good as $y$, and every alternative not considered is strictly worse than $y$.

\begin{defn}
A regular data set $(\gB,C)$ satisfies the {\bf Sequentially Strong Algebraic Axiom of Revealed Preference (S-SAARP)} if there is $\Gamma: \gB\rightrightarrows X$ such that $C(B)\subseteq \Gamma(B)\subseteq B$ and

\begin{itemize}
	\item [(i)] $(\gB,\Gamma)$ satisfies WAARP with respect to $\gF_1$, and
	\item [(ii)] $(\Gamma(\gB), C)$ is a regular data set satisfies SAARP with respect to $\gF_2$.
\end{itemize}
\end{defn}

Note that the axiom above is computationally feasible. The reason is that we only care whether other chosen points belong to $\Gamma$, and to control for that we can use integer variables. 	While all tests presented  previously can be implemented with linear programming techniques, the S-SAARP can be implemented as a mixed integer program with no more than $|\gB|(|\gB|-1)$ integer variable. Finally, note that we not only require $(\Gamma(\gB), C)$ to satisfy SAARP but also to be regular.
	The latter requirement implicitly imposes the monotonicity as being one of the requirements to be consistent with ordered theory.

\begin{cor}
\label{cor:TwoStageApplication}
A regular data set $(\gB,C)$ is sequentially rationalizable if and only if it satisfies S-SAARP.
\end{cor}

\subsubsection*{Good Enough Procedure}
\label{par:GoodEnoughApplication}

We adopt the idea of good enough choice as the chosen alternative being among the $k$-maximal alternatives in the set.
This theory also assumes that underlying preference relation is complete and transitive.
Hence, we assume that $\gF$ is an ordered theory.

Let $X$ be a (no more than countable) set of alternatives, let $1 <k <\vert X \vert$ be the integer number which defines the $k$-maximal alternatives (number can be specific for each budget) to be chosen from.
The logic is that alternatives chosen are among the ``good enough'' set.
For every budget $B$ let
$$ 
M^1_{R}(B) = \argmax\limits_{R} B
\text{ and } 
M^t_{R}(B) = M^{t-1}\cup \argmax_{R} (B\setminus M^{t-1}_R(B)).
$$
That is $M^k(B)$ contains $k$-maximal alternatives determining the ``good enough'' set.
Note that every alternative in $M^k(B)$ is strictly better than any alternative outside the $M^k(B)$.

\begin{defn}
A data set is {\bf ``good enough'' rationalizable} if there is a complete and transitive binary relation $R$ consistent with ordered theory $\gF$ such that there is $\Gamma(B) \supseteq M^{k}_R(B)$ such that 
$$ 
C(B) = \max (\Gamma(B)\cup C(B), R)
$$
for every $B\in \gB$.
\end{defn}

\noindent
Before we proceed further, we introduce a supplementary construct.
A data set $(\bar \gB,\bar C)$ is a \emph{truncated data set} generated by $(\gB,C)$ if for every $B\in \gB$ there is is set $B^{\downarrow}\subseteq B\setminus C(B)$ such that $\vert B\setminus B^{\downarrow}\vert \leq k$ and for every $B\in \gB$ there are 
\begin{itemize}
	\item [--] $\bar B^0 = B^{\downarrow}\cup C(B)\in \bar \gB$,
	\item [--] $\bar B^y = B^{\downarrow}\cup\{y\}\in \bar \gB$ for every $y\in B\setminus (B^{\downarrow}\cup C(B))$
\end{itemize}
such that
\begin{itemize}
	\item [--] $\bar C(\bar B^0) = C(B)$ for every $\bar B^0\in \bar \gB$,
	\item [--] $\bar C(\bar B^y) = y$ for every $\bar B^y\in \bar \gB$.
\end{itemize}

\noindent
The idea behind the extended experiment is that the $B^{\downarrow}$ refers to the $\Gamma(B)$, that is the consideration set.
The  requirement on cardinality of $B^{\uparrow}$ makes sure that $\Gamma(B) \supseteq M^{k(B)}_R(B)$.

\begin{defn}
A data set $(\gB,C)$ satisfies the {\bf ``Good Enough'' Strong Algebraic Axiom of Revealed Preference (GE-SAARP)} if there is a truncated regular data set $(\bar \gB,\bar C)$ generated by $(\gB,C)$ that satisfies SAARP.
\end{defn}

In terms of computational complexity, we can implement GE-SAARP as a mixed integer program as in the previous case.
	The number of integer variables would not exceed $\vert\gB\vert(\vert \gB\vert-1)$, since we only need to track whether the (transformation) of every chose point is in $B^{\uparrow}$ or not. 
	Finally,  note that we require the truncated data set to be regular.
	That is another example how regularity becomes a part of the testable restrictions.

\begin{cor}
\label{cor:GoodEnoughApplication}
A data set is ``good enough'' rationalizable if and only if it satisfies GE-SAARP
\end{cor}

\section{Concluding remarks}

\noindent
	We provide a comprehensive algebraic approach to revealed preference.
	This allows us to construct a revealed preference test which serves as a criteria for every data set to be rationalizable with a theory which satisfies group structure.
	We show that this revealed preference axiom unifies a variety of existing axioms, including those for transitive \citep[see][]{afriat1967,diewert1973afriat,forges2009,nishimura2017comprehensive}, homothetic \citep[see][]{varian1983,heufer2013testing,heufer2019homothetic}, and quasilinear preferences \citep[see][]{brown2007nonparametric,castillofreer2016}, as well as preferences that satisfy the independence axiom \citep[see][]{demuynck2009nash,polisson2015}.
	In addition, we show how the general results we provide can be used to construct tests for behavioral theories.
	We do so by constructing generalized revealed preference tests to the theories of sequential decision making and ``good enough'' choice.

	The version of SAARP we provide for preferences that satisfy the independence axiom shows that SAARP may (in some instances) take the form of a linear program.
	Linear programming is frequently used because of its computational tractability.
	If a test cannot be implemented using a linear program it often becomes intractable, and a researcher may only hope for very approximate solution.
	Therefore, characterizing the set of theories which can be tested using linear programming can be a fruitful avenue for further research based on the foundation this paper provides.

	Another fruitful set of applications comes from the revealed preference analysis of the group behavior.
	There has been recently some interest in applying revealed preference to social choice theory, e.g.\ \citet{duggan2016limits,duggan2019weak}; voting theory, e.g.\ \citet{tasos} and \citet{gomberg2018revealed},  and game theoretic models, e.g.\ \citet{brown1996testable}, \cite{echenique2013}, \cite{cherchye2013nash,cherchye2017household}, and \cite{castillo2019nonparametric}. 
	The unified algebraic approach may help to advance the progress along these lines, allowing to evaluate the set of questions which can be tackled with revealed preference theory.

\appendix
\section{Proofs for Section \ref{sec:Results}}

\setcounter{table}{0}
\renewcommand{\thetable}{A.\arabic{table}}
\setcounter{figure}{0}
\renewcommand{\thefigure}{A.\arabic{figure}}
\setcounter{lemma}{0}
\renewcommand{\theprop}{A.\arabic{prop}}
\setcounter{prop}{0}
\renewcommand{\thelemma}{A.\arabic{lemma}}
\setcounter{defn}{0}
\renewcommand{\thedefn}{A.\arabic{defn}}
\setcounter{cor}{0}
\renewcommand{\thecor}{A.\arabic{cor}}

Before we proceed with the proofs we need to introduce some auxiliary notation. 
Denote by $\gR$ the space of preference relations (reflexive binary relations).
Let $F:\gR\rightarrow \gR$ be a \emph{theory closure} corresponding to theory $\gF$, that is $(x,y)\in F(R)$ if there is $f\in \gF$ such that $(f(x),f(y))\in R$. We would omit an explicit reference to the theory when this can be done without confusion.
Let $T:\gR\rightarrow \gR$ be a \emph{transitive closure} that is $(x,y)\in T(R)$ if there is a finite sequence $x=s_1,\ldots,s_n=y$ such that $(s_j,s_{j+1})\in R$ for every $j\in \{1,\ldots,n-1\}$.
The idea behind the closures is that they are constructive coutnerpart of the consistency of preference relation with given assumption.
Hence, the important construct is the fixed point of the closure, i.e. $R=F(R)$ and $R=T(R)$.

Before we proceed to showing the connection between the constructed fucntions and some formal properties which would show that these function are actually closures.
In particular, we are interested in the three following properties.
A function $G:\gR\rightarrow \gR$ is said to be \emph{increasing} if $R\subseteq G(R)$.
A function $G:\gR\rightarrow \gR$ is said to be \emph{monotone} if  $R\subseteq R'$ implies $G(R)\subseteq G(R')$.
A function $G:\gR\rightarrow \gR$ is said to be \emph{idempotent} if $G(G(R))=G(R)$.
A function $G:\gR\rightarrow \gR$ is said to be a \emph{closure} if it is increasing, monotone and idempotent.

\begin{lemma}[\cite{demuynck2009}]
\label{lemma:TransitiveClosure}
$T:\gR\rightarrow \gR$ is a closure.
\end{lemma}

\begin{lemma}
\label{lemma:TheoryClosure}
$F:\gR\rightarrow\gR$ is a closure.
\end{lemma}

\begin{proof}
Following the definition of the closure we are going to prove the three properties which define it.

\bigskip

\noindent
{\bf $F$ is increasing.}
\\
Recall that $I\in \gF$. Hence, for every $(x,y)\in R$ then letting $f=I = I^{-1}$, we know that $f(x,y)\in R$, and therefore $(x,y)\in F(R)$.
\newline

\noindent
{\bf $F$ is monotone.}
\\
Take $R\subseteq R'$, $(x,y)\in F(R)$ if there is $f\in \gF$ such that $(f(x),f(y)) \in R\subseteq R'$.
Hence, $(x,y)\in F(R')$ and therefore, $F(R)\subseteq F(R')$.

\bigskip

\noindent
{\bf $F$ is idempotent.}
\\
Note that from the fact that $F$ is increasing we know that $F(R)\subseteq F(F(R))$.
Hence, we are left to show that $F(F(R))\subseteq F(R)$.
Consider $(x,y)\in F(F(R))$, then there is $f\in\gF$ such that $(f(x),f(y))\in F(R)$.
Then, in its order there is $f'\in \gF$ such that $(f'(f(x)),f'(f(y)))\in R$.
Given that $(\gF,\circ)$ is a group, then there is $\hat f = f'\circ f\in \gF$ and $(\hat f(x),\hat f(y))\in R$, and therefore, $(x,y)\in F(R)$.

\end{proof}

\noindent
Next, let us show that the fixed point of the corresponding closure (theory or transitive) exhibit the properties reflected in the assumption.
The result we provide for the transitive closure is straight-forward and based in the previous literature.

\begin{lemma}[\cite{demuynck2009,freer2016representation}]
\label{lemma:TransitiveFixedPoint}
$T(R)$ is transitive, i.e. $(x,y),(y,z)\in T(R)$ implies $(x,z)\in T(R)$ for every $x,y,z\in X$.
\end{lemma}

\noindent
The result for theory closure is stronger than one would think we need.
Recall that consistency with the theory implies that $(x,y)\in R$ implies $(f(x),f(y))\in R$ for every $x,y\in X$ and $f\in \gF$.
We show the equivalence instead of implication as being needed in the further proofs.  
That is, $(x,y)\in R$ if and only if $(f(x),f(y))\in R$ for every $x,y\in X$ and $f\in \gF$

\begin{lemma}
\label{lemma:TheoryFixedPoint}
$(x,y)\in F(R)$ if and only if $(f(x),f(y))\in F(R)$ for every $x,y\in X$ and $f\in \gF$.
\end{lemma}

\begin{proof}
{\bf  ($\Rightarrow$)}
Since $(x,y)\in F(R)$, then there are $(z,w)\in R$ such that $z =\bar f(x)$ and $w = \bar f(y)$ for some $\bar f\in \gF$.
Since $(\gF,\circ)$ is a group, then there are $f^{-1},\bar f^{-1}\in \gF$ and $\hat f = [\bar f^{-1} \circ f^{-1}]\in \gF$.
Hence, we can express $x=\bar f^{-1}(z)$ and $y=\bar f^{-1}(w)$.
Then, $\hat f(f(x)) = \bar f^{-1} (x) = z$ and $\hat f(f(y)) = \bar f^{-1} (x) = w$.
Hence, by construction of $F$ we can conclude that $(f(x),f(y))\in F(R)$.

\bigskip

\noindent
{\bf  ($\Leftarrow$)}
Consider $(f(x),f(y))\in F(R)$, then there is $(\bar f(f(x)),\bar f(f(y)))\in R$ by construction of $F$.
Since $(\gF,\circ)$ is a group, then there is $\hat f = \bar f \circ f \in \gF$.
Hence, $(\hat f(x),\hat f(y))\in R$ and therefore $(x,y)\in F(R)$ by construction of $F$.

\end{proof}

\noindent
Next, we investigate the connection between the closures constructed.
In particular, we show that every fixed point of 

\begin{lemma}
\label{lemma:TransitiveTheoryFixedPoint}
Let $R=F(R)$, then $F(T(R))=T(R)$.
\end{lemma}

\begin{proof}
Since $F$ is increasing, then $T(R)\subseteq F(T(R))$.
Hence, we are left to show the opposite, i.e. $F(T(R)) \subseteq T(R)$.
Take $(x,y)\in F(T(R))$, then there is $f\in \gF$, such that $(f(x),f(y))\in T(R)$.
That is there is a sequence $f(x) = s_1,\ldots, s_n = f(y)$ such that
$$
(s_j,s_{j+1}) \in R \text{ for every } j\in \{1,\ldots,n-1\}.
$$
Given that $R=F(R)$ and $(\gF,\circ)$ is a group, then there is $f^{-1}\in \gF$ such that 
$$
(f^{-1}(s_j),f^{-1}(s_{j+1}) ) \in R \text{ for every } j\in \{1,\ldots,n-1\}.
$$
Hence, $(f^{-1}(s_1),f^{-1}(s_n)) = (x,y)\in T(R)$.

\end{proof}

\noindent
Next we define a \emph{revealed preference relation} denoted by $R_E$.
That is, $(x,y)\in R_E$ if $x\in C(B), \ y\in B$ for some $B\in \gB$ or $x=y$.
	The first part defines actually the revealed preference part of the relation, while the second part guarantees  reflexivity.
Let $R\preceq R'$ ($R'$ is an {\bf extension} of $R$) if $R\subseteq R'$ and $P(R)\subseteq P(R')$.  
Let us start with introducing the operationalizable version of the same definition.

\begin{lemma}[\cite{demuynck2009,freer2016representation}]
\label{lemma:Consistency}
Let $R\subseteq R'$. $R\preceq R'$ if and only if $P^{-1}(R)\cap R'=\emptyset$.
\end{lemma}

\noindent
Given the extension nomenclature we can restate the conditions for the (weak) rationalization of the data.
In order to find a preference relation which generates the choice, we need to find $R^*$ such that $R_E\preceq R^*$.
This condition in particular guarantees that chosen points are maximal.
Let us note that we assume that choice correspondence contains \emph{all} best points in the budget and therefore, just requiring $R_E\subseteq R^*$ is not enough.
To guarantee that the binary relation $R^*$ is consistent with the theory we need to guarantee that $R^*=F(R^*)$.
To guarantee that the binary relation $R^*$ is transitive we need to guarantee that $R^*=T(R^*)$.
Hence, the rationalizability conditions can be restated as follows.
A data set is 
\begin{itemize}
	\item [--] \emph{weakly rationalizable} if there is a binary relation $R^*=F(R^*)$ such that $R_E\preceq R^*$;
	\item [--] \emph{rationalizable} if there is a binary relation $R^*=T(F(R^*))$ such that $R_E\preceq R^*$;
	\item [--] \emph{completely weakly rationalizable} if there is a complete binary relation $R^*=F(R^*)$ such that $R_E\preceq R^*$;
	\item [--] \emph{completely rationalizable} if there is a complete binary relation $R^*=T(F(R^*))$ such that $R_E\preceq R^*$.
\end{itemize}

\subsection{Proof of Proposition \ref{prop:WeakRationalization}}
\begin{proof}
{\bf $(\Rightarrow)$}
Let us proceed proving this statement by contradiction.
That is a data set is weakly rationalizable, though fails WAARP,
That is, there are $x_i\in C(B_i)$ and $x_i \in C(B_j)$ such that 
$$
f(x_i)\in B_j \text{ and } f^{-1}(x_j)\in B_i\setminus C(B_i).
$$
Since the data set is rationalizable and $f(x_i)\in B_j$, then $(x_j,f(x_i))\in R^*$.
Since $R^*$ is consistent with the theory and $(\gF,\circ)$ is a group, then $(f^{-1}(x_j),[f^{-1}\circ f](x_i)) = (f^{-1}(x_j),x_i)\in R^*$.
At the same time $f^{-1}(x_j)\in B_i\setminus C(B_i)$ implies that $(x_i,f^{-1}(x_j))\in P(R^*)$.
That is a contradiction.

\bigskip

\noindent
{\bf $(\Leftarrow)$}
Assume that WAARP is satisfied.
Let us start with showing that $R_E\preceq F(R_E)$.
One the contrary let us assume that $(z,w) \in P^{-1}(R_E)\cap F(R_E)$.
By construction there is $f\in \gF$ such that $(f(z),f(w))\in R_E$.
That is, $f(z) = x_j\in C(B_j)$ and $f(w)\in B_j$.
At the same time $(w,z)\in P(R_E)$ implies that $w=x_i\in C(B_i)$ and $f^{-1}(x_j)\in B_i\setminus C(B_i)$.
That is an immediate contradiction to WAARP.
Hence, $R_E\preceq F(R_E)$.
Lemma \ref{lemma:TheoryFixedPoint} implies that $F(R_E)$ is consistent with the theory and $R_E\preceq F(R_E)$ implies that $C(B)=\max(B, F(R_E))$.

\bigskip

\noindent
Let us conclude by providing a formal argument for the latter point.
Consider $x\in C(B)$, then $(x,y)\in R_E$ for every $y\in B$ and $(y,x)\in R_E$ only if $y\in C(B)$.
Since $R_E\preceq F(R_E)$, then $(x,y)\in R_E$ for every $y\in B$ since $R_E\subseteq F(R_E)$; and $(y,x)\in R_E$ only if $y\in C(B)$ since $P(R_E)\subseteq P(F(R_E))$.
\end{proof}

\subsection{Proof of Proposition \ref{prop:Rationalization}}
\begin{proof}
{\bf $(\Rightarrow)$}
We assume on the contrary that the data is rationalizable and there is a violation of SAARP.
That is there are sequences $x_1, \ldots, x_n$ such that $x_j\in C(B_j)$ for every $ j\in\{1,\ldots, n\}$  for some $B_1,\ldots, B_n\in \gB$, and $f_1,\ldots,f_{n-1}\in \gF$ such that
$$
f_{j}(x_{j+1}) \in B_j \text{ for every } j\in\{1,\ldots, n-1\},
$$
and 
$$ 
\left[\Comp\limits_{j=1}^{n-1} f_j\right]^{-1} (x_1) \in B_n\setminus C(B_n).
$$

Since the data set is rationalizable, there is a monotone preference relation $R^*$ transitive and consistent with the theory $\gF$ such that 
$$
(x,y) \in R \text{ for every } x\in C(B); \ y\in B \text{ for some } B\in \gB.
$$
Hence, $f_{j}(x_{j+1}) \in B_j$ implies that 
$$
(x_j,f_j(x_{j+1}))\in R^* \text{ for every } j\in {1,\ldots,n-1}.
$$
However, to show the relation between $x_1$ and $x_n$ we need to make sure that the the less preferred element in the pair $j$ is the more preferred element in the pair $j+1$.
Consider the first two pairs in this sequence, namely
$(x_1,f_1(x_2))$ and $(x_2,f_2(x_3))$.
Given that $R^*$ is consistent with the theory, then $(x_2,f_2(x_3))\in R^*$ implies that $(f_1(x_2),[f_1\circ f_2](x_3))\in R^*$.
Extending the same logic we can conclude that 
$$
(x_j,f_j(x_{j+1}))\in R^* \text{ implies } \left(\left[\Comp_{i=1}^{j-1} f_i \right] (x_j), \left[\Comp_{i=1}^{j} f_i \right] (x_{j+1}) \right) \in R^*
$$
for every $j\in \{2,\ldots, n-1\}$.

Then, by transitivity we can conclude that 
$$
\left(x_1,\left[\Comp_{i=1}^{n-1} f_i \right] (x_{n})\right) \in R^*.
$$
Given that $(\gF,\circ)$ is a group, there is an inverse $\left[\Comp_{i=1}^{n-1} f_i \right]^{-1}\in \gF$ and given that $R^*$ is consistent with $\gF$ we can conclude that 
$$
\left(\left[\Comp_{i=1}^{n-1} f_i \right]^{-1}(x_1), x_{n}\right) \in R^*.
$$
On the other hand the violation of AARP implies that 
$$
\left[\Comp\limits_{j=1}^{n-1} f_j\right]^{-1} (x_1) \in B_n\setminus C(B_n),
$$
then the fact that $R^*$ is a monotone preference relation that rationalizes the data set implies
$$ 
\left(x_{n},\left[\Comp_{i=1}^{n-1} f_i \right]^{-1}(x_1)\right) \in P(R^*).
$$
That is a direct contradiction given that by construction $(x,y) \in P(R^*)$ only if $(x,y)\in R^*$ and $(y,x)\notin R^*$, i.e. $P(R^*)$ denotes the strict (asymmetric) part of the preference relation $R^*$.

\bigskip

\noindent
{\bf $(\Leftarrow)$}
Since data set satisfies SAARP, then it satisfies WAARP.
Therefore, as we have proven above $R_E\preceq F(R_E)$.
Hence, to complete the proof (given Lemma \ref{lemma:TransitiveTheoryFixedPoint}) we need to show that $R_E\preceq T(F(R_E))$ if the data satisfies SAARP.
We proceed with the proof by contradiction.
That is, assume that there is $(x,y)\in P^{-1}(R_E)\cap T(F(R_E))$.
Then, there is a sequence of $x=s_1,\ldots, s_n=y$ such that $(s_j,s_{j+1})\in F(R_E)$ for every $j\in \{1,\ldots,n-1\}$. 
Since $(s_j,s_{j+1})\in F(R_E)$, then there is $f_j\in \gF$ such that $(f_j(s_j),f_j(s_{j+1}))\in R_E$.
Hence, we can restate the sequence that adds $(x,y)$ to $T(F(R_E))$ as follows.
There are sequences  $x=s_1,\ldots, s_n=y$  and $f_1,\ldots,f_{n-1}\in \gF$ such that 
$$
(f_j(s_j),f_j(s_{j+1})) \in R_E \text{ for every } j\in\{1,\ldots,n-1\}.
$$
Given the construction of the revealed preference relation we know that $f_j(s_j) = x_j\in B_j$ for some $B_j\in \gB$.
Moreover, $f_j(s_{j+1}) = [f_j\circ f^{-1}_{j+1}](x_{j+1}) \in B_j$, where $f^{-1}_{j+1}, [f_j\circ f^{-1}_{j+1}] \in \gF$ since $(\gF,\circ)$ is a group.
Denote by $\hat f_j = [f_j\circ f^{-1}_{j+1}]$ for every $j\in \{1,\ldots,n-2\}$ and $\hat f_{n-1}= f_{n-1}$.
Then, we can rewrite the sequence as 
$$
\hat f_j (x_{j+1})\in B_j \text{ for every } j\in \{1,\ldots,n-1\}.
$$
At the same time we know that $(y,x)\in P(R_E)$, that is
$$
y = s_n = x_n \in C(B_n) \text{ and } x=f^{-1}_1(x) \in B_n\setminus C(B_n).
$$
Let us compute the
\begin{align*}
\left[ \Comp_{j=1}^{n-1} \hat f_j \right]^{-1} = \left[ [f_1\circ f^{-1}_2]\circ [f_2\circ f^{-1}_3]\circ \ldots \circ[f_{n-2}\circ f^{-1}_{n-1}]\circ [f_{n-1}] \right]^{-1} = f_1^{-1}.
\end{align*}
Therefore, the following holds
$$
\left[ \Comp_{j=1}^{n-1} \hat f_j \right]^{-1}(x_1) = f_1^{-1}(x_1) \in B_n\setminus C(B_n)
$$
that is a direct contradiction of SAARP.

\end{proof}

\subsection{Proof of Proposition \ref{prop:CompleteWeakRationalization}}
Necessity of WAARP for complete weak rationalization obviously follows from the necessity of WAARP for the weak rationalization.
Hence, we are left to provide the sufficiency proof.
For the sufficiency proof we will need couple auxiliary properties of the closures and an auxiliary result.
A function $G:\gR\rightarrow \gR$ is said to be \emph{algebraic} if for any $R\in\gR$ and all $(x,y)\in G(R)$ there is a finite relation $R'\subseteq R$ such that $(x,y)\in G(R')$.
That is, to for any comparison in $G(R)$ there is a finite sub-relation to add this comparison. 
Denote by $N(R) = X\times X \setminus (R\cup R^{-1})$
A function $G:\gR\rightarrow \gR$ is said to be \emph{weakly expansive} if for any $R=G(R)$ such that $N(R)\ne \emptyset$, then there is a non-empty $S\subseteq N(R)$ such that $R\cup S\preceq F(R\cup S)$.
These two properties allow us to restate the result from \cite{demuynck2009}.

\begin{lemma}[\cite{demuynck2009} Extension Theorem]
\label{lemma:DemuynckExtensionThm}
Let $G:\gR\rightarrow \gR$ be a weakly expansive, algebraic closure and let $R_E$ be a revealed preference relation.
There is a complete binary relation $R^*=G(R^*)$ such that $R_E\preceq R^*$ if and only if $R_E\preceq F(R_E)$.
\end{lemma}

\noindent
Lemma \ref{lemma:DemuynckExtensionThm} already guarantees us the sufficiency proof once we manage to show that $F$ is weakly expansive and algebraic.
Recall that the condition of $R_E\preceq F(R_E)$ is equivalent to WAARP, and existence of complete binary relation $R^*=F(R^*)$ is equivalent to the existence of complete and consistent with the theory and guarantees that the observed choices are the best in the given budget set.
Hence, to complete the proof we show that $F$ is algebraic and weakly expansive.

\begin{lemma}
\label{lemma:ExpansiveClosure}
A theory closure $F$ is algebraic and weakly expansive.
\end{lemma}

\begin{proof}
{\bf $F$ is algebraic}
\\
Consider a relation $R$ and an element $(x,y)\in F(R)$, then there is $f\in \gF$ such that $(f(x),f(y))\in R$.
Let $D=\{f(x),f(y)\}$ and let $R'=R\cap(D\times D)$.
Then, $(x,y)\in F(R')$ and $R'$ is finite by definition.

\bigskip

\noindent
{\bf $F$ is weakly expansive}
\\
Consider $R=F(R)$ and $(x,y)\in N(R)$.
Let $R' = R\cup\{(x,y)\}$ and let us show that $R'\preceq F(R')$.
On the contrary assume that $(z,w)\in F(R')\cup P^{-1}(R')$.
Let us start from showing that $(z,w)=(f(x),f(y))$ for some $f\in \gF$.
On the contrary assume that there are $(x,y)\neq (x',y')\in R'$, such that $z=f(x')$ and $w=f(y')$.
Then, $(z,w)\in R$ and therefore, $(w,z)\notin P(R')$ since $P(R')$ is the asymmetric (strict) part of $R'$.

Next, we are going to rely on the fact that $(z,w)=(f(x),f(y))$.
Since $(w,z)\in P(R') \subseteq R\cup\{(x,y)\}$, then either (i) $f=I$ and $(y,x)\in R$ that contradicts the fact that $(x,y)\in N(R)$, or $(f(y),f(x))\in R$.
The latter statement (given that $(\gF,\circ)$ is a group) implies that there is $f^{-1}\in \gF$, and therefore $([f^{-1}\circ f(y)],[f^{-1}\circ f](x)) = (y,x) \in R$ that is a contradiction to $(x,y)\in R$.

\end{proof}

\subsection{Proof of Proposition \ref{prop:CompleteRationalization}}
Since we have undergone the significant change to the notion of the theory and notion of consistency. 
We need to modify the notion of theory closure to account for these changes.
Let $\bar F:\gR\rightarrow \gR$ be an \emph{ordered theory closure}, if $(x,y)\in \bar F(R)$ if there are $\bar f \leq f\in \gF$ such that $(\bar f(x),f(y))\in R$.
Hence, to proceed we need to follow the scheme that to the large extent repeats the proofs provided before.
\begin{enumerate}
	\item $T(\bar F):\gR\rightarrow \gR$ is an algebraic closure
	\item $T(\bar F(R))$ is transitive and consistent with the ordered theory
	\item $T(\bar F):\gR\rightarrow \gR$ is weakly expansive.
	\item SAARP (on the regular data set) is equivalent to $R_E\preceq T(\bar F(R_E))$.
\end{enumerate}

\bigskip

\noindent
{\bf (1) $T(\bar F):\gR\rightarrow \gR$ is a closure}

\noindent
Let us note that a composition of closures is a closure as well.
Hence, given that we already shown that $T$ is a closure (see Lemma \ref{lemma:TransitiveClosure}), we are left to show that $\bar F$ is a closure to complete this part.

\begin{lemma}
\label{lemma:AlgebraicOrderedClosure}
$\bar F:\gR\rightarrow \gR$ is an algebraic closure.
\end{lemma}

\begin{proof}
{\bf $\bar F$ is increasing.}
\\
Since $I\in\gF$, then $(x,y)\in R$ implies that $(x,y)\in \bar F(R)$.

\bigskip

\noindent
{\bf $\bar F$ is monotone.}
\\
Let $R\subseteq R'$ and assume on the contrary that there is $(x,y)\in \bar F(R)$ and $(x,y)\notin \bar F(R')$.
Since $(x,y)\in \bar F(R)$, then there are $\bar f\le f\in \gF$ such that $(\bar f(x),f(y))\in R\subseteq R'$.
Latter implies that$ (\bar f(x),f(y))\in R'$ and therefore $(x,y)\in \bar F(R')$.

\bigskip

\noindent
{\bf $\bar F$ is idempotent.}
\\
Since $\bar F$ is increasing and monotone we already know that $\bar F(R)\subseteq \bar F(\bar F(R))$.
Hence, we are left to show that $\bar F(\bar F(R)) \subseteq \bar F(R)$.
Consider $(x,y)\in \bar F(\bar F(R))$, then there are $\bar f\leq f\in \gF$ such that $(\bar f(x),f(y))\in  \bar F(R)$.
Then, there are $\bar f'\leq f' \in \gF$ such that $([\bar f'\circ \bar f](x),[f'\circ f](y))\in R$.
Recall that $\gF$ is a group, therefore, $[\bar f'\circ \bar f], [f'\circ f]\in \gF$.
Moreover, $\gF$ is an ordered group, and therefore, $[\bar f'\circ \bar f]\leq [f'\circ f]$.
Then, there are $\bar f^* = [\bar f'\circ \bar f] \leq  [f'\circ f] = f^* $ such that $(\bar f^*(x),f^*(y))\in R$ and therefore, $(x,y)\in \bar F(R)$.

\bigskip

\noindent
{\bf $\bar F$ is algebraic}
\\
Consider $(x,y)\in \bar F(R)$, then there are $\bar f\le f\in \gF$ such that $(\bar f(x),f(y))\in R$.
Hence, the finite relation can be $D=\{(\bar f(x),f(y))\} \subseteq R$ and $(x,y)\in \bar F(D)$.

\end{proof}

\noindent
Let us note that a finite composition of algebraic closures is also algebraic closure.
Even though, this is a rather straight-forward observation, we provide the formal proof in Lemma \ref{lemma:Closures}.
Hence, knowing that $T$ and $\bar F$ are both algebraic closures, we conclude that $T\circ \bar F$ is an algebraic closure as well.

\bigskip

\noindent
{\bf (2) $T(\bar F(R))$ is transitive and consistent with the ordered theory}

\begin{lemma}
\label{lemma:OrderedConsistentFixedPoint}
$R=\bar F(R)$ is consistent with ordered theory.
\end{lemma}

\begin{proof}
Suppose on the contrary that $R=\bar F(R)$, but $R$ is not consistent with the theory. 
Hence, there are $\bar f\geq f\in \gF$ such that $(\bar f(x),f(y))\notin R$.
Let $z = \bar f(x)$ and $w = f(y)$, then (given that $(\gF,\circ)$ is a group), there are $\bar f^{-1}, f^{-1} \in \gF$ such that $\bar f^{-1}\le f^{-1}$ (given that $(\gF,\circ,\ge)$ is an ordered group).
Hence, $(\bar f^{-1}(z),f^{-1}(w))\in R=\bar F(R)$ and $\bar f^{-1}\le f^{-1}$ implies that $(z,w)\in \bar F(R)=R$.
The latter observation is equivalent to $(\bar f(x),f(y))\in R$, that is a contradiction.

\end{proof}

\begin{lemma}
\label{lemma:OrderedClosureTransitivity}
Let $R=F(R)$, then $T(R) = \bar F(T(R))$
\end{lemma}

\begin{proof}
Since $\bar F$ is increasing then $T(R)\subseteq \bar F(T(R))$ and we are left to show that $\bar F(T(R))\subseteq T(R)$.
Assume $(x,y)\in \bar F(T(R))$, then there are $\bar f\le f\in \gF$ such that $(\bar f(x),f(y))\in T(R)$.
Then, there is a sequence $\bar f(x) = s_1,\ldots, s_n=f(y)$ such that
$$
(s_j,s_{j+1})\in R \text{ for every } j\in \{1,\ldots, n-1\}.
$$

\noindent
Since $(\gF,\circ,\ge)$ is an ordered group, then there are $\bar f^{-1},f^{-1}\in \gF$ and $f^{-1}\ge f^{-1}$.
Since $R=F(R)$ is consistent with ordered theory $\gF$ (see Lemma \ref{lemma:OrderedConsistentFixedPoint}) and $(\bar f(x),s_2)\in R$, then
$$
([\bar f^{-1}\circ \bar f](x),f^{-1}(s_1)) = (x,f^{-1}(s_1))\in R.
$$

\noindent
Moreover, for $2\le j \le n-1$ consistency of $R$ implies $(f^{-1}(s_{j}),f^{-1}(s_{j+1}))\in R$.
Finally, $f^{-1}(s_n) = [f^{-1}\circ f](y)$, therefore, $(x,y)\in T(R)$.

\end{proof}

\bigskip

\noindent
{\bf (3) $T(\bar F):\gR\rightarrow \gR$ is weakly expansive.}

\begin{lemma}
\label{lemma:TransitiveOrderedWeaklyExpansiveClosure}
$T(\bar F): \gR \rightarrow \gR$ is weakly expansive.
\end{lemma}

\begin{proof}
We consider $R=T(\bar F(R))$ such that $N(R)\ne \emptyset$.
Let $(x,y)\in N(R)$ and $\bar R = R\cup \{(x,y)\}$.
Suppose on the contrary to the weak expansiveness that there is $(z,w)\in T(\bar F(\bar R))$ and $(w,z)\in P(\bar R)$.
Since we are interested in the composition of the closures, we need to proceed proving the weak expansiveness sequentially. 
We start from making the claim that inner closure is weakly expansive, with the proceeding similarly to one of Lemma \ref{lemma:ExpansiveClosure}.

\begin{claim}
\label{claim:ExpansiveInnerClosure}
$\bar F:\gR\rightarrow \gR$ is weakly expansive.
\end{claim}

\begin{proof}[Proof of Claim \ref{claim:ExpansiveInnerClosure}]
Suppose on the contrary that there is $(z,w)\in \bar F(\bar R)$ and $(w,z)\in P(\bar R)$.
We start from showing that $(z,w)=(\bar f(x),y)$ for some $\bar f\le f\in\gF$.
Suppose the contrary, that is there is $(\bar x,\bar y)\ne (z,w)\in \bar R$ such that $(\bar f(\bar x),f(\bar y))\in \bar R$.
Since, $\bar R = R\cup\{(x,y\}$, then $(\bar f(\bar x),f(\bar y))\in R$ then $(z,w)\in F(R)=R$, that is a contradiction to the fact that $(w,z)\in P(R)$ as $P(\cdot)$ being the asymmetric part of the relation.

Hence, $(z,w)=(\bar f(x),f(y))$ for some $\bar f\le f\in\gF$.
Let us note that $f\ne I$, otherwise we would obtain an immediate contradiction.
Recall that $(\gF,\circ,\ge)$ is an ordered group, then there are $\bar f^{-1}\ge f^{-1}\in \gF$.
Since $(w,z)\in P(\bar R)$, then $(\bar f^{-1}(w),f^{-1}) = (y,x) \in R$ since $R$ is consistent with the theory.
That is contradiction to the fact that $(x,y)\in N(R)$.

\end{proof}

\bigskip

\noindent
Given Claim \ref{claim:ExpansiveInnerClosure} we know that $(z,w)\notin \bar F(\bar R)$.
Hence, $(z,w)\in T(\bar F(\bar R))$ implies there is a non-trivial sequence of $z=s_1,\ldots, s_n=y$ such that 
$$
(\bar f_j(s_j),f_j(s_{j+1}))\in \bar R \text{ for some } \bar f_j\le f_j \text{ for every } j\in\{1,\ldots,n-1\}.
$$
Let us consider one of the \emph{shortest sequences} that add $(z,w)$, i.e. such that there is no shorter sequence.
To complete the proof we need to show that there is exactly one entry of $(x,y)$ to this sequence.
This claim would significantly simplify our proof as the rest of the comparisons would belong to $R$.

\begin{claim}
\label{claim:XYUniqueEntry}
There is unique $k\le n-1$ such that $(\bar f_k(s_k),f_k(s_{k+1})) =(x,y)\in \bar R$.
\end{claim}

\begin{proof}[Proof of Claim \ref{claim:XYUniqueEntry}]
We start from illustrating why there should be at least one entry of $(\bar f_k(s_k),f_k(s_{k+1}) =(x,y)$.
If there is none, then $(\bar f_j(s_j),f_j(s_{j+1})\in R$ for every $j$ and therefore, $(z,w)\in R=T(\bar F(R))$.
That is a contradiction to the fact that $(w,z)\in P(R)$.
Given that there is at least one entry of $(x,y)$, we next proceed with proving that this entry should be unique.
Consider on the contrary that there are at least two entries and assume without loss of generality that $k<l$ and correspond to the first and second entry of the $(x,y)$ to the sequence.
Since the $(\gF,\circ,\geq)$ is ordered group, we can proceed by considering two cases.

\medskip

\noindent
{\bf Case 1: $f_k\le \bar f_l$.}
Given that $f_k\le\bar f_l $ and $\bar f_l \le f_l$, we can conclude that $f_k\le f_l$.
Recall that by assumption $f_k(s_{k+1})=y=f_l(s_{l+1})$.
Given that $(\gF,\circ,\ge)$ is an ordered group, there are $f_{k}^{-1}\ge f^{-1}_l\in\gF$.
Since $R$ is reflexive and consistent with the theory, then $(y,y)\in R$ implies that $(f_{k}^{-1}(y),f^{-1}_{l}(y)) =(s_{k+1},s_{l+1})\in R\subseteq \bar R$.
That is the very same sequence can be expressed without using the second entry of the $(x,y)$.
That is a contradiction.

\medskip

\noindent
{\bf Case 1: $f_k\le \bar f_l$.}
Since $k$ and $l$ correspond to the first and second entry of $(x,y)$, then $(f_k(s_{k+1}),\bar f_l(s_l))\in R$ as $R=T(\bar F(R))$.
Moreover, $f_k(s_{k+1})=y$ and $\bar f_l(s_l)=x$, therefore, $(y,x)\in R$ that is a contradiction to $(x,y)\in N(R)$.

\end{proof}

\bigskip

\noindent
To complete the proof let us appeal to the Claim \ref{claim:XYUniqueEntry}, we can ensure that there is a unique entry of $(x,y)$ to the sequence.
Given that $(w,z)\in R$, we can conclude that $(s_{k+1}, s_k)\in R$.
Given that $f_k\ge \bar f_k$ and $R$ is consistent with the theory, then $(y,x) = (f_k(s_{k+1}),\bar f_k(s_k))\in R$.
That is a contradiction to the fact that $(x,y)\in N(R)$.
Hence, $T\circ \bar F$ is weakly expansive.

\end{proof}

\bigskip

\noindent
{\bf (3) SAARP is equivalent to $R_E\preceq T(\bar F(R_E))$. }

\begin{lemma}
\label{lemma:SAARP2OrderedTheory}
If a data set satisfies SAARP, then $R_E\preceq T(\bar F(R_E))$.
\end{lemma}

\begin{proof}
Suppose on the contrary that there is $(z,w)\in T(\bar F(R_E))$ and $(w,z) \in P(R_E)$.
Since $(z,w)\in T(\bar F(R_E))$, then there is a sequence $z=s_1,\ldots, s_n$ such that for every $j\in\{1,\ldots,n-1\}$ there are $\bar f_{j}\le f_j\in \gF$ such that
$$
(\bar f_j(s_j),f_j(s_{j+1})) \in R_E.
$$
Since $(\bar f_j(s_j),f_j(s_{j+1})) \in R_E$, the the construction of preference relation implies that $\bar f_j(s_j) = x_j\in C(B_j)$ for some $B_j\in\gB$.
Recall that $(\gF,\circ,\ge)$ is ordered group, then there are $f_j^{-1}\le \bar f_j^{-1}\in\gF$.
Hence, $f_j(s_{j+1}) = [f_j\circ\bar f^{-1}_{j+1}](x_{j+1})\in B_j$, then $[f_j\circ f^{-1}_{j+1}](x_{j+1})\in B_j$ given that data set is regular and $[f_j\circ\bar f^{-1}_{j+1}](x_{j+1})\le [f_j\circ f^{-1}_{j+1}](x_{j+1})$.
Hence, denoting by $\hat f_j = [f_j\circ f^{-1}_{j+1}]$ for every $j\le n-1$ and $\hat f_{n-1} = f_{n-1}$ we obtain 
$$
\hat f_j(x_{j+1}) \in B_j \text{ for every } j\in \{1,\ldots, n-1\}.
$$
Then,
\begin{align*}
\left[ \Comp_{j=1}^{n-1} \hat f_j \right]^{-1} = \left[ [f_1\circ f^{-1}_2]\circ [f_2\circ f^{-1}_3]\circ \ldots \circ[f_{n-2}\circ f^{-1}_{n-1}]\circ [f_{n-1}] \right]^{-1} = f_1^{-1}.
\end{align*}
Recall that $f^{-1}\le \bar f^{-1}$ and $\bar f^{-1}(x_1)\in B_{n}\setminus C(B_n)$, then $\bar f^{-1}(x_1)\in B_{n}\setminus C(B_n)$ since data set is regular. 
Therefore, the following holds
$$
\left[ \Comp_{j=1}^{n-1} \hat f_j \right]^{-1}(x_1) = f_1^{-1}(x_1) \in B_n\setminus C(B_n)
$$
that is a contradiction of SAARP.

\end{proof}

\begin{proof}[Proof of Proposition \ref{prop:CompleteRationalization}]
{\bf ($\Rightarrow$)}
Necessity proof follows argument similar to one in proof of Proposition \ref{prop:Rationalization}. 
Therefore, we leave this an exercise to the reader.

\bigskip

\noindent
{\bf ($\Rightarrow$)}
Lemma \ref{lemma:SAARP2OrderedTheory} guarantees that if data set satisfies SAARP, then $R_E\preceq T(\bar F(R_E))$.
Then, given that we have shown that $T(\bar F(R_E))$ is weakly expansive algebraic closer, we can apply Demuynck Extension Theorem (see Lemma \ref{lemma:DemuynckExtensionThm}) to guarantee that there is a complete $R^*$ such that $R_E\preceq R^* =T(\bar F(R^*))$.
Since, $T(\bar F(R^*))=R^*$ it is consistent with ordered theory and transitive (see Lemmas \ref{lemma:OrderedConsistentFixedPoint} and \ref{lemma:OrderedClosureTransitivity}).
Given that $R_E\preceq R^*$ then, $R^*$ rationalizes the data set.

\end{proof}

\section{Proofs for Section \ref{sec:Applications} }

\setcounter{lemma}{0}
\renewcommand{\theprop}{B.\arabic{prop}}
\setcounter{prop}{0}
\renewcommand{\thelemma}{B.\arabic{lemma}}
\setcounter{defn}{0}
\renewcommand{\thedefn}{B.\arabic{defn}}
\setcounter{cor}{0}
\renewcommand{\thecor}{B.\arabic{cor}}

\noindent
Since Corollaries \ref{cor:TransitiveApplication}, \ref{cor:HomotheticApplication}, \ref{cor:QuasilinearApplication} and \ref{cor:IndependentApplication} immediately follows from Propositions \ref{prop:Rationalization} and \ref{prop:CompleteRationalization} we omit these proofs. 
Instead we concentrate on the proofs for Corollaries \ref{cor:TwoStageApplication} and \ref{cor:GoodEnoughApplication}.

\subsection{Proof of Corollary \ref{cor:TwoStageApplication}}
\begin{proof}
{\bf ($\Rightarrow$)}
Consider there is a violation of S-SAARP, then for every $C(B)\subseteq\Gamma(B)\subseteq B$ either (i) $(\gB,\Gamma)$ violates WAARP or (ii) $(\Gamma(\gB),C)$ violates SAARP.
If $(\gB,\Gamma)$ then as result of Proposition \ref{prop:WeakRationalization} there is no binary relation $R_1$ that is consistent with the theory $\gF_1$ such that $\Gamma(B) = \max(B,R_1)$, that is a contradiction.
If $(\Gamma(\gB),C)$ violates SAARP, then there is no complete and transitive binary relation consistent with ordered theory $\gF_2$ such that $C(B)=\max(\Gamma(B),R_2)$, that is a contradiction.

\bigskip

\noindent
{\bf ($\Leftarrow$)}
To proceed with sufficiency assume that S-SAARP is satisfied.
Then, there is $C(B)\subseteq\Gamma(B)\subseteq B$, such that (i) $(\gB,\Gamma)$ satisfies WAARP and (ii) $(\Gamma(\gB),C)$ satisfies SAARP.
Since $(\gB,\Gamma)$ satisfies WAARP, then there is a binary relation $R_1$ that is consistent with $\gF_1$ and $\Gamma(B) = \max(B,R_1)$ due to Proposition \ref{prop:WeakRationalization}.
Since $(\Gamma(\gB),C)$ satisfies SAARP being a regular data set, then Proposition \ref{prop:CompleteRationalization} implies that there is a complete and transitive binary relation consistent with the theory $\gF_2$ such that $C(B)=\max(\Gamma(B),R_2)$.

\end{proof}

\subsection{Proof of Corollary \ref{cor:GoodEnoughApplication}}
\begin{proof}
{\bf ($\Rightarrow$)}
Assume that data set is good enough rationalizable.
Then, given the complete and transitive binary relation $R^*$ that is consistent with the theory, there are $\Gamma(B)\supseteq M^k_R(B)$ such that
$$
C(B) = \max(\Gamma(B)\cup C(B),R^*).
$$

\noindent
Let $\bar B^0 = \Gamma(B)\cup C(B)$, and $\bar B^y = \{y\}\cup \Gamma(B)$ for every $y\in B\setminus (\Gamma(B)\cup C(B))$, where $\bar C(\bar B^0)=C(B)$ and $\bar C(\bar B^y) = y$.
Since $\Gamma(B)\supseteq M^k_R(B)$, then $y$ is better than every alternative in $\Gamma(B)$, that is
$$
y = \max(\bar B^y,R^*).
$$
This collection generates a truncated data set $(\bar \gB,\bar C)$.
Then, Proposition \ref{prop:CompleteRationalization} implies that $(\bar \gB,\bar C)$ satisfies SAARP.
Hence, GE-SAARP is satisfied as we just constructed the regular truncated data set that satisfies SAARP.

\bigskip

\noindent
{\bf ($\Leftarrow$)}
Assume that GE-SAARP is satisfied.
Then, there is a complete and transitive binary relation $R^*$ that is consistent with ordered theory $\gF$ and chosen points are maximal.
Given that we modified the underlying experiment we need to construct a different revealed preference relation.
Note that $(x,y)\in \bar R_E$ if $x\in C(\bar B)$ and $y\in \bar B$ for some $\bar B\in \bar \gB$.
Let 
$$
B^{\uparrow} = B\setminus B^{\downarrow} \text{ for every } B\in \gB.
$$

\noindent
Note that by construction $(x,y) \in \bar R_E$, we know that every $x\in B^{\uparrow}$ and $y\in B^{\downarrow}$.
Note that by construction $\vert B^{\uparrow}\vert\le k$.
Recall that by construction $\bar R_E\preceq R^*$, that in according to the extension, every alternative from $B^{\uparrow}$ is better than every alternative form $B^{\downarrow}$.
Hence, we can conclude that $\Gamma(B) = B^{\downarrow} \subseteq M^k_{R^*}(B)$.

\end{proof}

\section{Auxiliary Results}

\setcounter{lemma}{0}
\renewcommand{\theprop}{C.\arabic{prop}}
\setcounter{prop}{0}
\renewcommand{\thelemma}{C.\arabic{lemma}}
\setcounter{defn}{0}
\renewcommand{\thedefn}{C.\arabic{defn}}
\setcounter{cor}{0}
\renewcommand{\thecor}{C.\arabic{cor}}

\noindent
There are two groups of auxiliary results we need.
First group are the results about the compositions of closures.
That is, we show that some papers are transfered to the composition of the closures.
Second group are the basic properties of the ordered theories used in the proofs above.

\begin{lemma}
\label{lemma:Closures}
Let $F_1,F_2:\gR\rightarrow\gR$ are algebraic closures, then $F = F_2\circ F_1:\gR\rightarrow\gR$ is increasing, monotone and algebraic.
\end{lemma}

\begin{proof}
{\bf $F$ is increasing.}
\\
If $(x,y)\in R$, then $(x,y)\in F_1(R)$ since $F_1$ is increasing.
If $(x,y)\in F_1(R)$, then $(x,y)\in F_2(F_1(R))$, since $F_2$ is increasing.
Hence, $(x,y)\in R$, then $(x,y)\in F(R) = F_2(F_1(R))$.

\bigskip

\noindent
{\bf $F$ is monotone.}
\\
If $R\subseteq R'$, then $F_1(R)\subseteq F_1(R')$ since $F_1$ is monotone.
Given $F_1(R)\subseteq F_1(R')$, then $F_2(F_1(R))\subseteq F_2(F_1(R'))$ since $F_2$ is monotone.
Hence, $R\subseteq R'$ implies that $F_2(F_1(R)) = F(R) \subseteq F(R') = F_2(F_1(R'))$ since $F_2$.

\bigskip

\noindent
{\bf $F$ is algebraic.}
\\
Consider $(x,y)\in F(R) = F_2(F_1(R))$.
Given that $F_2$ is algebraic then there is a finite $D\subseteq F_1(R)$ such that $(x,y)\in F_2(F_1(D))$.
Given that $D$ is finite, then for every $(x_i,y_i)\in D\subseteq F_1(R)$ there is $D_i$ such that $(x_i,y_i)\in F_1(D_i)$, given that $F_1$ is algebraic as well.
Hence, let 
$$
\bar D = \bigcup\limits_{(x_i,y_i)\in D} D_i
$$ 
then by construction $(x,y)\in F_2(F_1(\bar D))$.
Therefore, $F=F_2\circ F_1$ is algebraic.

\end{proof}

\begin{lemma}
\label{lemma:Algebra}
Let $(\gF,\circ,\ge)$ be ordered group, then 
\begin{itemize}
	\item [(i)] $f\ge f'$ and $\bar f\ge \bar f'$ implies $f\circ\bar f\ge f'\circ\bar f'$,
	\item [(ii)] $f\ge \bar f$ implies $f^{-1}\le \bar f^{-1}$.
\end{itemize}
\end{lemma}

\begin{proof}
{\bf Part (i).}
\\
If $f\ge f'$, then $f\circ \bar f\ge f'\circ \bar f$, since  $(\gF,\circ,\ge)$ is an ordered group.
If $\bar f\ge \bar f'$ then $f'\ge \bar f\ge f'\circ \bar f'$, since  $(\gF,\circ,\ge)$ is an ordered group.
Then,
$$
f\circ \bar f\ge f'\circ \bar f\ge f'\circ \bar f' \Rightarrow f\circ \bar f \ge f'\circ \bar f'.
$$

\bigskip

\noindent
{\bf Part (ii).}
\\
Since $(\gF,\circ,\ge)$ is an ordered group and $f\ge\bar f$ we can conduct the following reasoning
\begin{align*}
f\ge \bar f 
& \Rightarrow f^{-1}\circ f\ge f^{-1}\circ \bar f 
\Rightarrow I \ge f^{-1}\circ \bar f
\Rightarrow \\ 
& \Rightarrow I\circ \bar f^{-1} \ge f^{-1}\circ \bar f \circ \bar f^{-1}
\Rightarrow \bar f^{-1}\ge f^{-1}.
\end{align*}

\end{proof}

\clearpage

\bibliographystyle{plainnat}
\bibliography{GRP.bib}

\end{document}